\documentclass[11pt]{article}
\usepackage[letterpaper,margin=1in]{geometry}
\usepackage{amssymb,amsmath,amsthm}
\usepackage{thmtools,thm-restate}
\usepackage{enumerate}
\usepackage[shortlabels]{enumitem}
\usepackage{linegoal}
\usepackage[T1]{fontenc}
\usepackage{caption}
\usepackage{subcaption}
\usepackage[svgnames]{xcolor}
\usepackage{complexity}
\usepackage{cite}
\usepackage{float}
\usepackage{comment}
\usepackage{dirtytalk}

\usepackage[numbers]{natbib}
\usepackage[colorlinks,
linkcolor={blue},
citecolor={black},
breaklinks=true,
urlcolor={blue}]{hyperref}
\usepackage[capitalize]{cleveref}

\usepackage{algorithm}
\usepackage[noend]{algpseudocode}
\newcounter{tbox}

\setlist[itemize]{leftmargin=5.5mm}
\newcommand{\defn}[1]{\textcolor{Maroon}{\emph{#1}}}

\newtheorem{ltheorem}{Theorem}

\theoremstyle{plain}
\newtheorem{theorem}{Theorem}
\newtheorem{lemma}[theorem]{Lemma}

\theoremstyle{definition}

\newcommand*{\myproofname}{Proof}

\usepackage{tikz}
\usetikzlibrary{arrows}
\usetikzlibrary{arrows.meta}
\usetikzlibrary{decorations.pathmorphing, decorations.pathreplacing, decorations.shapes}
\usetikzlibrary{decorations.markings}
\usetikzlibrary{calc}
\usetikzlibrary{shapes.geometric}

\tikzset{
  circ/.style = {circle,draw,fill,inner sep=1.1pt},
  circg/.style = {circle,draw=lightGray,fill=lightGray,inner sep=1.1pt},
  invisible/.style = {circle,draw=none,inner sep=0pt,font=\tiny},
  nonedge/.style={decorate,decoration={snake,amplitude=.3mm,segment length=1mm},draw}
}

\usepackage{todonotes}

\usepackage{booktabs,tabularx}

\ExplSyntaxOn

\NewDocumentCommand{\problemStatement}{mm}
 {
  \arteche_problemstatement:nn { #1 } { #2 }
 }

\prop_new:N \l_arteche_problemstatement_body_prop

\cs_new_protected:Nn \arteche_problemstatement:nn
 {
  \prop_set_from_keyval:Nn \l_arteche_problemstatement_body_prop { #2 }
  \begin{center}
  \begin{tabularx}{\columnwidth}{@{}lX@{}}
  \toprule
  \multicolumn{2}{@{}c@{}}{\textsc{#1}}\tabularnewline
  \midrule
  \prop_map_function:NN \l_arteche_problemstatement_body_prop \__arteche_problemstatemet_do:nn
  \bottomrule
  \end{tabularx}
  \end{center}
}

\cs_new_protected:Nn \__arteche_problemstatemet_do:nn
 {
  \bfseries \tl_rescan:nn { } { #1 }: & #2 \\
 }

\ExplSyntaxOff

\usepackage{tcolorbox}
\tcbuselibrary{skins}

\newtcolorbox{mybox}[1]{minipage boxed title*=-2cm,
enhanced,attach boxed title to top center=
{yshift=-3mm,yshifttext=-1mm},colback=RoyalBlue!5!white,
boxed title style={size=small,colback=RoyalBlue!20},coltitle=black,
center title,title={#1}}

\title{Odd Cycle Transversal on $H$-free graphs}
\date{}

\author{
	Esther Galby\thanks{Chalmers University of Technology, Sweden.}
	\and
Paloma T. de Lima\thanks{IT University of Copenhagen, Denmark, and Norwegian School of Economics, Norway. Supported by the Independent Research Fund Denmark grant agreement number 2098-00012B.}
	\and
	Andrea Munaro\thanks{University of Parma, Italy.}
	\and
	Amir Nikabadi\thanks{IT University of Copenhagen, Denmark. Supported by the Independent Research Fund Denmark grant agreement number 2098-00012B.}
}

\begin{document}
 \maketitle

\begin{abstract}
\textsc{Odd Cycle Transversal} is a classic $\mathsf{NP}$-hard graph optimization problem asking for a minimum-weight set of vertices whose deletion makes the input graph bipartite, or equivalently, a maximum-weight induced bipartite subgraph. We show that \textsc{Odd Cycle Transversal} is quasi-polynomial-time solvable on $kP_4$-free graphs, for every fixed $k \in \mathbb{N}$. In fact, we provide an $n^{O_k(\log n)}$-time algorithm for the more general \textsc{Max-Weight List $2$-Colorable Induced Subgraph}, where the notation $O_{k}(\cdot)$ hides factors depending on $k$. Paired with known results from the literature, this allows us to obtain a complete complexity dichotomy for these two problems on $H$-free graphs into cases solvable in quasi-polynomial time and cases which are $\mathsf{NP}$-hard, in particular resolving an open problem of Agrawal, Lima, Lokshtanov, Saurabh, and Sharma [SODA 2024]. 

Our algorithms are based on a new structural tool that may be of independent interest. We introduce the notion of $H$-amiable family and show that, for every fixed graph $H$ without isolated vertices and every fixed $k\ge2$, every $kH$-free graph admits an $H$-amiable family of quasi-polynomial size that can be constructed in quasi-polynomial time. Besides yielding the aforementioned algorithms, this result gives, for every fixed connected graph $H$ and every fixed $k\ge2$, a reduction from \textsc{Max-Weight Independent Set} on $kH$-free graphs to the same problem on $H$-free graphs with $n^{O_{H,k}(\log n)}$ overhead. In this setting, it improves the $n^{O_{H,k}(\log^3 n)}$ overhead obtained by specializing the general reduction of Gartland and Lokshtanov [FOCS 2020].
\end{abstract}

\section{Introduction}

\textsc{Odd Cycle Transversal} is a classic and well-studied graph optimization problem that, given a vertex-weighted graph, asks to find a minimum-weight subset of vertices whose deletion results in a bipartite graph. By complementing the solution, this problem is equivalent to finding a maximum-weight induced bipartite (i.e., $2$-colorable) subgraph of the input graph.   

In this paper, we study a fundamental class of graph optimization problems generalizing \textsc{Odd Cycle Transversal}. A problem in this class consists in finding a maximum-weight induced subgraph satisfying a certain fixed property $\Pi$. A graph property $\Pi$ is \defn{nontrivial} if it is true for infinitely many graphs and false for infinitely many graphs, and it is \defn{hereditary} if, whenever a graph satisfies $\Pi$, all its induced subgraphs satisfy $\Pi$ as well. A classic result of \citet*{LY80} states that, whenever the fixed property $\Pi$ is nontrivial and hereditary, the corresponding problem is $\mathsf{NP}$-hard. In particular, \textsc{Odd Cycle Transversal} is $\mathsf{NP}$-hard. We are interested in the case where $\Pi$ is the property of being list $r$-colorable. In order to properly define the corresponding problem, we first recall some definitions. 

Let $G = (V, E)$ be a finite simple graph. A \defn{coloring} of $G$ is a mapping $\varphi \colon V\rightarrow \{1, 2, \ldots\}$ that gives each vertex $u \in V$ a \defn{color} $\varphi(u)$ in such a way that, for every two adjacent vertices $u$ and $v$ in $G$, we have $\varphi(u) \neq \varphi(v)$. For $r \geq 1$, a coloring $\varphi$ of $G$ is an \defn{$r$-coloring} if $\varphi(u) \in \{1, \ldots, r\}$ for every $u \in V$, and a graph is \defn{$r$-colorable} if it admits an $r$-coloring. For $r \geq 1$, an \defn{$r$-list assignment} of $G$ is a function $\mathcal{L}\colon V \rightarrow 2^{\{1,\ldots,r\}}$ that assigns each vertex $u \in V$ a \defn{list} $\mathcal{L}(u) \subseteq \{1,\ldots,r\}$ of admissible colors for $u$. A coloring $\varphi$ of $G$ \defn{respects} $\mathcal{L}$ (or is an \defn{$\mathcal{L}$-coloring}) if  $\varphi(u)\in \mathcal{L}(u)$ for every $u\in V$. For a fixed integer $r \geq 1$, the problem we are interested in is the following.

\begin{mybox}{\textsc{Max-Weight List $r$-Colorable Induced Subgraph}}
\textbf{Input:} A graph $G$ with a weight function $w\colon V(G) \rightarrow \mathbb{Q}_{\geq 0}$ and an $r$-list assignment $\mathcal{L}$ of $G$.\\
\textbf{Task:} Find a subset $F \subseteq V(G)$ such that the induced subgraph $G[F]$ admits a coloring that respects $\mathcal{L}$, and the weight $w(F) = \sum_{v \in F}w(v)$ is maximum subject to this condition.
\end{mybox}

All rational weights are encoded in binary. Throughout the paper, a running-time bound for a weighted problem that is expressed only in terms of the number $n$ of vertices suppresses a factor polynomial in the total binary encoding length of the input weights.

\textsc{Max-Weight List $r$-Colorable Induced Subgraph}, besides generalizing \textsc{Odd Cycle Transversal}, is a common generalization of several other well-known and deeply investigated $\mathsf{NP}$-hard problems. For example, for $r=1$, \textsc{Max-Weight List $r$-Colorable Induced Subgraph} is equivalent to \textsc{Max-Weight Independent Set}, which is the problem of finding a maximum-weight subset of pairwise non-adjacent vertices of a given graph. Moreover, \textsc{Max-Weight List $r$-Colorable Induced Subgraph} generalizes \textsc{List $r$-Coloring} and hence \textsc{$r$-Coloring} as well, which are known to be polynomial-time solvable for $r \leq 2$ (see, e.g., \citep{Tuz97}) and \textsf{NP}-hard for all $r > 2$ \citep{K72}. Here, for a fixed $r \geq 1$, \textsc{List $r$-Coloring} is the problem of deciding whether a given graph $G$ with an $r$-list assignment $\mathcal{L}$ admits a coloring that respects $\mathcal{L}$. By setting $\mathcal{L}(u)=\{1,\ldots,r\}$ for every $u\in V(G)$, we obtain {\sc $r$-Coloring}. Note also that, for $r_1 \leq r_2$, \textsc{List $r_1$-Coloring} is a special case of \textsc{List $r_2$-Coloring}.

Given the hardness of \textsc{Max-Weight List $r$-Colorable Induced Subgraph}, a considerable effort has been made to delineate the boundary of (in)tractability for restricted classes of inputs, most notably for \defn{hereditary classes}, that is, graph classes closed under vertex deletion, which are particularly well suited for a systematic investigation. The ultimate goal is to obtain complexity dichotomies telling us for which hereditary classes the problem can or cannot be solved efficiently (under standard complexity assumptions). Before reviewing the state of the art of this rich area at the intersection of structural and algorithmic graph theory, we need to recall some definitions. 

A graph $G$ is \defn{$H$-free}, for some graph $H$, if it contains no induced subgraph isomorphic to $H$, that is, we cannot modify $G$ into $H$ by a sequence of vertex deletions. For a family of graphs $\mathcal{F}$, a graph is \defn{$\mathcal{F}$-free} if it is $H$-free for every $H \in \mathcal{F}$, and a class is hereditary if and only if it coincides with the class of $\mathcal{F}$-free graphs for some (possibly infinite) $\mathcal{F}$. The \defn{disjoint union} $G + H$ of graphs $G$ and $H$ is the graph with vertex set $V(G) \cup V(H)$ and edge set $E(G) \cup E(H)$. We denote the disjoint union of $k$ copies of $G$ by $kG$ and let $P_s$ denote the chordless path on $s$ vertices. 

Henceforth, we focus on hereditary classes obtained by forbidding a single induced subgraph, a setting which turned out to be particularly interesting. It is known that \textsc{Odd Cycle Transversal} remains $\mathsf{NP}$-hard\footnote{Note that all hardness results mentioned in this paper hold even in the unweighted setting.} on $H$-free graphs whenever $H$ is not a linear forest (i.e., a disjoint union of paths). Indeed, \citet*{chiarelli2018minimum} showed that it is \textsf{NP}-hard on $H$-free graphs if $H$ contains a cycle or a claw (i.e., the $4$-vertex star). In recent years, considerable work has been done toward classifying the complexity of \textsc{Odd Cycle Transversal} and the more general \textsc{Max-Weight List $r$-Colorable Induced Subgraph} on graphs forbidding an induced linear forest, and obtaining complete complexity dichotomies ($\mathsf{P}$ or $\mathsf{QP}$ vs $\mathsf{NP}$-hard) on $H$-free graphs has been repeatedly posed as an open problem, for example in  \citep{agrawal2024odd,ACPSC25,CKPRS21,CPS19,GLMN25}. We now review the known results, which we then summarize in \Cref{thm:literature-oct}.

Consider first \textsc{Odd Cycle Transversal}. \citet*{chiarelli2018minimum} showed that it is polynomial-time solvable on $kP_2$-free graphs for all $k\in \mathbb{N}$. \citet*{dabrowski2020cycle} showed that it is polynomial-time solvable on $(P_3 + kP_1)$-free graphs for all $k\in \mathbb{N}$ and $\mathsf{NP}$-hard on $\{P_6, P_5+P_2\}$-free graphs. \citet*{ALLSS24} provided a quasi-polynomial-time algorithm on $P_5$-free graphs, which was later improved to a polynomial-time algorithm on the same graph class by \citet*{agrawal2024odd}. 

Consider now \textsc{Max-Weight List $r$-Colorable Induced Subgraph} for $r \geq 2$. \citet*{lokshtanov2024maximum} showed that it is polynomial-time solvable on $P_5$-free graphs. This was extended by \citet*{henderson2024maximum}, who showed that it is polynomial-time solvable on $(P_5 + kP_1)$-free graphs for all $k\in \mathbb{N}$. \citet*{GLMN25} showed that it is polynomial-time solvable on $kP_3$-free graphs for all $k\in \mathbb{N}$. Moreover, the problem remains $\mathsf{NP}$-hard on $(P_4+P_2)$-free graphs for all $r \geq 5$, since \textsc{List $r$-Coloring} is $\mathsf{NP}$-hard on this class, as shown by  \citet*{couturier2015list}.    

\begin{ltheorem}\label{thm:literature-oct}Let $r \geq 2$ be a fixed integer. \textsc{Max-Weight List $r$-Colorable Induced Subgraph} and \textsc{Odd Cycle Transversal} on $H$-free graphs are polynomial-time solvable if $H$ is an induced subgraph of either $kP_3$ or $P_5 + kP_1$, for some $k \geq 1$, and remain \textsf{NP}-hard if $H$ contains $P_6$ or $P_5+P_2$ as an induced subgraph. 
Moreover, for $r \geq 5$, \textsc{Max-Weight List $r$-Colorable Induced Subgraph} on $H$-free graphs remains \textsf{NP}-hard if $H$ contains $P_4+P_2$ as an induced subgraph.
\end{ltheorem}

As observed in \citep{GLMN25}, \Cref{thm:literature-oct} gives a complete complexity dichotomy ($\mathsf{P}$ vs $\mathsf{NP}$-hard) for \textsc{Max-Weight List $r$-Colorable Induced Subgraph} on $H$-free graphs when $r \geq 5$. Thus the only open cases toward a complete dichotomy for $r \geq 2$ arise when $r \in \{2,3,4\}$ and $H$ is an induced subgraph of $kP_4$, for $k \in \mathbb{N}$.  

\subsection{Our results}

In this paper, we make considerable progress toward a complete complexity dichotomy ($\mathsf{P}$ vs $\mathsf{NP}$-hard) for \textsc{Max-Weight List $2$-Colorable Induced Subgraph} and \textsc{Odd Cycle Transversal} on $H$-free graphs. Our main result is the following. The notation $O_{k}(\cdot)$ hides factors depending on $k$.

\begin{theorem}\label{thm:main_1}
For every $k \in \mathbb{N}$, \textsc{Max-Weight List $2$-Colorable Induced Subgraph} admits an $n^{O_k(\log n)}$-time algorithm on $n$-vertex $kP_4$-free graphs.
\end{theorem}

\Cref{thm:main_1} provides strong evidence that \textsc{Max-Weight List $2$-Colorable Induced Subgraph} is not $\mathsf{NP}$-hard on $kP_4$-free graphs, as
otherwise every problem in $\mathsf{NP}$ would be solvable in quasi-polynomial time (this would contradict, for example, the Exponential-Time Hypothesis). As an immediate consequence of \Cref{thm:main_1}, we obtain the following. 

\begin{theorem}\label{thm:main}
For every $k \in \mathbb{N}$, \textsc{Odd Cycle Transversal} admits an $n^{O_k(\log n)}$-time algorithm on $n$-vertex $kP_4$-free graphs.
\end{theorem}

Combining \Cref{thm:main_1,thm:main} with \Cref{thm:literature-oct}, we obtain a \textit{complete} complexity dichotomy for \textsc{Max-Weight List $2$-Colorable Induced Subgraph} and  \textsc{Odd Cycle Transversal} on $H$-free graphs into cases solvable in quasi-polynomial time and cases which are $\mathsf{NP}$-hard. This resolves, in particular, \citep[Problem~4, p.~5288]{ALLSS24}. More precisely, the following holds.

\begin{theorem} Both \textsc{Max-Weight List $2$-Colorable Induced Subgraph} and \textsc{Odd Cycle Transversal} on $H$-free graphs are polynomial-time solvable if $H$ is an induced subgraph of either $kP_3$ or $P_5+kP_1$ for some $k \geq 1$, quasi-polynomial-time solvable in all remaining cases where $H$ is an induced subgraph of $kP_4$ for some $k \geq 1$, and $\mathsf{NP}$-hard otherwise.
\end{theorem}

\begin{proof}
If $H$ is not a linear forest, then \textsc{Odd Cycle Transversal} is $\mathsf{NP}$-hard on $H$-free graphs by the result of \citet*{chiarelli2018minimum} recalled above. Since \textsc{Odd Cycle Transversal} is a special case of \textsc{Max-Weight List $2$-Colorable Induced Subgraph}, the same conclusion holds for the latter problem.

Suppose that $H$ is a linear forest. If every component of $H$ has at most four vertices, then $H$ is an induced subgraph of $kP_4$ for some $k\geq1$. In this case, the polynomial-time cases follow from \Cref{thm:literature-oct}, and all remaining cases are quasi-polynomial-time solvable by \Cref{thm:main_1,thm:main}. We may therefore assume that $H$ has a component on at least five vertices. If $H$ has a component on at least six vertices, then $H$ contains $P_6$ as an induced subgraph. Otherwise, $H$ has a component isomorphic to $P_5$. If every other component is isolated, then $H$ is an induced subgraph of $P_5+kP_1$ for some $k\geq1$; in every other case, $H$ contains $P_5+P_2$ as an induced subgraph. The hardness statements now follow from \Cref{thm:literature-oct}. These cases exhaust all possibilities.
\end{proof}

The proof of \Cref{thm:main_1} extends the approach used in \citep{GLMN25} to obtain a polynomial-time algorithm for the more general \textsc{Max-Weight List $r$-Colorable Induced Subgraph} ($r \geq 2$) on the more restricted class of $kP_3$-free graphs. The key ingredient in the proof of such result from \citep{GLMN25} is that every $kP_3$-free graph $G$ admits an amiable family of size polynomial in $|V(G)|$ and which can be computed in polynomial time. Here, an \defn{amiable family} of $G$ (a notion first introduced by \citet*{LM12}) is a family $\mathcal{S} \subseteq 2^{V(G)}$ of subsets of $V(G)$ such that each member of $\mathcal{S}$ induces a $P_3$-free subgraph in $G$ and each independent set of $G$ is contained in some member of $\mathcal{S}$. Unfortunately, $kP_4$-free graphs are unlikely to admit even quasi-polynomial-time computable amiable families of quasi-polynomial size. Indeed, paired with \citep[Lemma~9]{GLMN25}, this would in particular give a quasi-polynomial-time algorithm for \textsc{Max-Weight List $5$-Colorable
Induced Subgraph} on $(P_4 + P_2)$-free graphs, which we already recalled is $\mathsf{NP}$-hard. However, the following relaxation of the notion of amiable family suffices for our purposes. For a graph $G$, a family
$\mathcal{S} \subseteq 2^{V(G)}$
is an \defn{$H$-amiable family} if it satisfies the following properties:
\begin{itemize}
    \setlength\itemsep{0em}
\item each member of $\mathcal{S}$ induces an $H$-free graph in $G$; 
\item each independent set of $G$ is contained in some member of $\mathcal{S}$.
\end{itemize}

The following result is one of the main technical contributions of the paper and is the key ingredient in the proof of \Cref{thm:main_1}, as we explain in~\Cref{sec:overview}. The notation $O_{H,k}(\cdot)$ hides factors depending on $H$ and $k$.

\begin{restatable}{theorem}{HAF}\label{thm:HAF}
For every fixed graph $H$ with no isolated vertices and every fixed integer $k \geq 2$, every $n$-vertex $kH$-free graph $G$ admits an $H$-amiable family of size
$n^{O_{H,k}(\log n)}$,
which can be computed in time
$n^{O_{H,k}(\log n)}$.
\end{restatable}

In related work,~\cite{nikabadi2026induced} introduced
induced-$\mathcal{H}$-packing treewidth. We observe that their Lemma~4.1, \say{container
lemma}, also
yields~\Cref{thm:HAF}. Indeed, apply that lemma with
$\mathcal{H}=\{H\}$, with distinguished set $B=V(G)$, and with
parameter $k-1$. Its packing hypothesis is satisfied, since $k$
pairwise anticomplete induced copies of $H$ in $G$ would induce
$kH$. The lemma computes, in time $n^{O_{H,k}(\log n)}$, a family
$\mathcal{P}$ of size $n^{O_{H,k}(\log n)}$ consisting of pairs
$(I,F)$ such that $I$ is independent,
$
F\subseteq V(G)\setminus N_G[I],
$
and $G[F]$ is $H$-free. Moreover, for every independent set
$I^\star$ of $G$, there exists $(I,F)\in\mathcal{P}$ such that
$
I\subseteq I^\star\subseteq I\cup F.
$
Consequently, the family
$
\{I\cup F:(I,F)\in\mathcal{P}\}
$
covers every independent set of $G$. For each
$(I,F)\in\mathcal{P}$, every vertex of $I$ is isolated in
$G[I\cup F]$. Since $H$ has no isolated vertices and $G[F]$ is
$H$-free, it follows that $G[I\cup F]$ is $H$-free. Thus these sets
form an $H$-amiable family. We nevertheless give a direct proof for~\Cref{thm:HAF} in~\Cref{sec:building-h-af}, which is much shorter in the present
$kH$-free setting.

\medskip

\Cref{thm:HAF} has another interesting application, as we explain next. Let $\mathcal{F}$ be the class of forests whose components are either paths or subdivided claws. It is well known that \textsc{Max-Weight Independent Set} on $H$-free graphs remains $\mathsf{NP}$-hard unless $H \in \mathcal{F}$, and classifying the complexity for $H \in \mathcal{F}$ is arguably one of the main open problems in algorithmic graph theory, dating back to the early 1980s. It is widely believed, and formally conjectured in \citep{GLMPPR24,L17}, that the problem is polynomial-time solvable on $H$-free graphs for every $H \in \mathcal{F}$. 
Note that every graph in $\mathcal{F}$ is an induced subgraph of $tS_{t,t,t}$ for some $t \in \mathbb{N}$; here, for $a,b,c \in \mathbb{N}$, $S_{a,b,c}$ is the graph obtained from the claw by subdividing its three edges $a-1$, $b-1$, and $c-1$ times, respectively.

\medskip

Despite a large number\footnote{We refer the reader to \citep{CMPPR24,GLMPPR24} for a summary of known results.} of polynomial-time algorithms for special $H \in \mathcal{F}$, the conjecture remains open and the best general result is a recent $n^{O_H(\log^{19} n)}$-time algorithm for every $H \in \mathcal{F}$ by \citet*{GLMPPR24}. This extended the earlier breakthrough of \citet*{GL20} that \textsc{Max-Weight Independent Set} admits an $n^{O_H(\log^{6} n)}$-time algorithm on $H$-free graphs when $H$ is the disjoint union of a path $P_t$ and $t-1$ copies of $S_{1,1,2}$, later simplified and improved to an $n^{O_H(\log^{5} n)}$-time algorithm by \citet*{PPR21}. A key ingredient underlying these algorithms is a reduction of \citet{GL20} from the case of arbitrary $H$ to the case where $H$ is connected, as we explain next. 

For a graph $H$, let $\mathsf{Or}(H)$ be a value oracle that takes as input an $H$-free graph $F$ together with a weight function $w_F\colon V(F)\to\mathbb{Q}_{\geq0}$ and returns the maximum weight of an independent set of $F$ with respect to $w_F$.
\citet*[Theorem~2]{GL20} provided an algorithm for \textsc{Max-Weight Independent Set} on $H$-free graphs that, given access to oracles $\mathsf{Or}(H_i)$ for all connected components $H_i$ of $H$, uses at most $n^{O_H(\log^3 n)}$ operations and oracle calls on induced subgraphs of the $n$-vertex input graph. This reduction is applied in \citep{GLMPPR24} to the case $H = tS_{t,t,t}$, for $t \in \mathbb{N}$. When the forbidden graph is $kH$ for a connected graph $H$, \Cref{thm:HAF} gives the following sharper reduction.

\begin{theorem}\label{oracle reduction}
Let $H$ be a connected graph and let $k\geq2$. Given an oracle $\mathsf{Or}(H)$, one can solve \textsc{Max-Weight Independent Set} on $kH$-free graphs using at most $n^{O_{H,k}(\log n)}$ additional operations and oracle calls on induced subgraphs of the $n$-vertex input graph.
\end{theorem}

\begin{proof}
Let $G$ be an $n$-vertex $kH$-free graph with weight function $w\colon V(G)\to\mathbb{Q}_{\geq0}$. If $H=K_1$, then every independent set of $G$ has size at most $k-1$, and a maximum-weight independent set can be found in $n^{O(k)}$ time by enumerating all vertex subsets of size at most $k-1$.

Assume that $|V(H)|\geq2$. Since $H$ is connected, it has no isolated vertices. By \Cref{thm:HAF}, we can compute an $H$-amiable family $\mathcal{S}$ of $G$ of size $n^{O_{H,k}(\log n)}$ within the same running time. For every $S\in\mathcal{S}$, a standard vertex-deletion self-reduction using $\mathsf{Or}(H)$ recovers a maximum-weight independent set $I_S$ of $G[S]$ with at most $|S|+1$ oracle calls on induced subgraphs of $G[S]$, which are $H$-free, each supplied with the restriction of $w$ to its vertex set. Return a set $I_S$ of maximum weight among these sets.

Every set $I_S$ is an independent set of $G$. Conversely, every independent set $I$ of $G$ is contained in some member $S$ of $\mathcal{S}$, and therefore $w(I_S)\geq w(I)$. Thus the returned set is a maximum-weight independent set of $G$. The additional factor of $n$ caused by the self-reduction is absorbed by $n^{O_{H,k}(\log n)}$, so the number of additional operations and oracle calls satisfies the claimed bound.
\end{proof}

\paragraph{Organization of the paper.} In~\Cref{sec:overview}, we give an overview of the proofs of our results. In \Cref{sec:prelim}, we introduce our notation and recall standard definitions. In \Cref{sec:building-h-af}, we prove \Cref{thm:HAF} by a direct branching argument. In \Cref{sec:algo}, we present the proof of \Cref{thm:main_1} in full detail. Finally, we conclude in \Cref{sec:conclusion} with some additional remarks and open problems.

\subsection{Technical overview}\label{sec:overview}

In this subsection, we provide a high-level description of our approach in proving our two main theorems. 

\paragraph{Amiable families (\Cref{thm:HAF}).} For an induced subgraph $F$ of the input graph, let $\mu(F)$ be the number of induced copies of $H$ in $F$. A vertex that belongs to an induced copy of $H$ receives as its score the number of copies of $H$ intersecting its closed neighborhood; all other vertices receive score zero. We branch on a vertex $v_F$ of maximum score. If $R$ induces $jH$ for the largest possible $j$, then $j\leq k-1$, and every induced copy of $H$ intersects $N_F[R]$, since otherwise it forms an induced copy of $(j+1)H$ together with $R$. Averaging over the at most $(k-1)|V(H)|$ vertices of $R$ shows that $N_F[v_F]$ intersects at least a $1/((k-1)|V(H)|)$ fraction of all induced copies of $H$.

The first branch continues in $F-v_F$ and covers every independent set avoiding $v_F$. The second branch continues in $F-N_F[v_F]$ and adds $v_F$ to each recursively generated set; it covers every independent set containing $v_F$, as all its remaining vertices lie outside $N_F[v_F]$. Since $v_F$ belongs to an induced copy of $H$, the first branch destroys at least one copy of $H$, whereas the second branch destroys a constant fraction of all copies of $H$. Consequently, if $T(m)$ is an upper bound on the number of leaves generated from a graph containing at most $m$ induced copies of $H$, then we have
\[
T(m)\leq T(m-1)+T\left(\left\lfloor\left(1-\frac{1}{(k-1)|V(H)|}\right)m\right\rfloor\right).
\]
This recurrence gives $T(m)\leq\exp(O_{H,k}(\log^2(m+1)))$. Every set generated by the first branch is $H$-free by induction. In the second branch, the added vertex has no neighbor in the recursively generated set; since $H$ has no isolated vertices, adding this vertex also preserves $H$-freeness. The recursion therefore constructs, in time
$n^{O_{H,k}(\log n)}$, an $H$-amiable family of size
$n^{O_{H,k}(\log n)}$.

\paragraph{\textsc{Max-Weight List $2$-Colorable Induced Subgraph} (\Cref{thm:main_1}).} The proof proceeds by reducing the problem to a sequence of maximum-weight independent set computations on perfect graphs. The key ingredient is a so-called \emph{two-layer auxiliary graph} that encodes list colorings as stable sets: given two subsets of vertices $A_1, A_2$ of a graph $G$, the auxiliary graph contains a copy of each vertex in $A_1$ ($A_2$, respectively) whose list contains color $1$ ($2$, respectively). For $i \in \{1,2\}$, the $i$-th layer of the auxiliary graph is isomorphic to $G[A_i\cap V_i(\mathcal{L})]$ to enforce that every color class is independent, and cross-edges are added in the auxiliary graph between two copies of the same vertex to prevent assigning two colors to one vertex. 

We first establish a weight-preserving bijection between feasible list colorings and stable sets of this auxiliary graph (see \Cref{lem:list2-assignment-correspondence-vertex}). The main structural step is then to prove that whenever the two subsets are restricted to induce $P_4$-free graphs, the auxiliary graph is perfect (see \Cref{lem:list2-aux-perfect}). The algorithm then uses the $P_4$-amiable family constructed in \Cref{thm:HAF} to enumerate quasi-polynomially many pairs of subsets, computes for each such pair a maximum-weight stable set in the corresponding auxiliary graph guaranteed to be perfect (using, e.g., the algorithm of \citet*{grotschel1981ellipsoid}) and finally outputs the largest found solution (see \Cref{alg:list2}). Finally, we prove that one of the enumerated pairs captures an optimal solution, implying correctness, while the quasi-polynomial runtime follows from the size of the amiable family and the polynomial work performed for each pair (see \Cref{thm:list2-kp4}). We refer to \Cref{sec:algo} for a complete proof of \Cref{thm:main_1}.

\section{Preliminaries}\label{sec:prelim}

We denote the set of positive integers by $\mathbb{N}$. For every $n\in \mathbb{N}$, we let $[n]:= \{1,\dots, n\}$.

Throughout the paper, graphs have finite vertex sets and no loops or parallel edges. Let $G$ be a~graph with vertex set $V(G)$ and edge set $E(G)$.
For $X \subseteq V(G)$, we denote the subgraph of $G$ \defn{induced} by $X$ as $G[X]$, that is $G[X] = (X, \{uv : u,v \in X \ \mbox{and} \ uv \in E(G)\})$, and we sometimes use $X$ to denote both the vertex subset of $G$ and the subgraph of $G$ it induces.

For a vertex $v \in V(G)$, we denote by $N_G(v)$ the open neighborhood of $v$ and by $N_G[v] := N_G(v) \cup \{v\}$ the closed neighborhood of $v$. For a set $X \subseteq V(G)$, we let
    $N_G(X) := \Bigl(\bigcup_{x \in X} N_G(x)\Bigr)\setminus X$ and
    $N_G[X] := N_G(X)\cup X$.
For disjoint sets $X, Y \subseteq V(G)$, we say that $X$ is \defn{complete} to $Y$ if every vertex in $X$ is adjacent to every vertex in $Y$, and $X$ is \defn{anticomplete} to $Y$ if there are no edges between $X$ and $Y$. A \defn{clique} in $G$ is a set of pairwise adjacent
vertices. A \defn{stable set} (or \defn{independent set}) in $G$ is a set of pairwise non-adjacent vertices.

For graphs $G$ and $G'$, an \defn{isomorphism} from $G$ to $G'$ is a bijection $\varphi \colon V(G) \to V(G')$ such that for all $x,y \in V(G)$, $xy \in E(G)$ if and only if $\varphi(x)\varphi(y) \in E(G')$; with a slight abuse of notation, we typically write $\varphi \colon G \to G'$ instead of $\varphi \colon V(G) \to V(G')$.


\section{Building $H$-amiable families}\label{sec:building-h-af}

This section is devoted to the proof of~\Cref{thm:HAF}.

\HAF*

\begin{proof}
Fix a graph $H$ with no isolated vertices and an integer $k\geq2$, and let $G$ be an $n$-vertex $kH$-free graph. We may assume that $H$ is nonempty. Let $h:=|V(H)|$; since $H$ has no isolated vertices, $h\geq2$. For every induced subgraph $F$ of $G$, let
\[
\mathcal{H}(F):=\{X\subseteq V(F):F[X]\simeq H\}
\qquad\text{and}\qquad
\mu(F):=|\mathcal{H}(F)|.
\]
For $v\in V(F)$, define
\[
s_F(v):=
\begin{cases}
|\{X\in\mathcal{H}(F):X\cap N_F[v]\neq\emptyset\}|,
& \text{if $v$ belongs to some member of $\mathcal{H}(F)$},\\[1mm]
0, & \text{otherwise}.
\end{cases}
\]

We recursively define a family $\Upsilon(F)\subseteq2^{V(F)}$. If $F$ is $H$-free, let $\Upsilon(F):=\{V(F)\}$. Otherwise, let $v_F$ be a vertex of maximum score and let
\[
\Upsilon(F):=\Upsilon(F-v_F)\cup
\bigl\{\{v_F\}\cup S:S\in\Upsilon(F-N_F[v_F])\bigr\}.
\]
Both recursive calls are made on proper induced subgraphs of $F$, and hence the recursion terminates.

We first establish how the measure $\mu$ decreases. Let $F$ be an induced subgraph of $G$ that is not $H$-free. Choose $j$ as large as possible such that $F$ contains an induced copy of $jH$, and let $R\subseteq V(F)$ induce such a copy. Since $F$ is $kH$-free,
$1\leq j\leq k-1$ and 
$|R|=jh\leq(k-1)h.$
Every $X\in\mathcal{H}(F)$ intersects $N_F[R]$. Indeed, if $X\cap N_F[R]=\emptyset$, then $X$ is disjoint from and anticomplete to $R$, and hence $F[R\cup X]\simeq(j+1)H$, contrary to the choice of $j$. Thus every member of $\mathcal{H}(F)$ is counted by $s_F(u)$ for at least one $u\in R$. Since every vertex of $R$ belongs to an induced copy of $H$, we obtain
\[
\mu(F)\leq\sum_{u\in R}s_F(u)
\leq |R|s_F(v_F)
\leq(k-1)h\,s_F(v_F).
\]
Consequently,
\[
s_F(v_F)\geq\frac{\mu(F)}{(k-1)h}.
\]
In particular, $s_F(v_F)>0$, so $v_F$ belongs to an induced copy of $H$ in $F$. Therefore
\[
\mu(F-v_F)\leq\mu(F)-1.
\]
The copies of $H$ that survive in $F-N_F[v_F]$ are exactly the members of $\mathcal{H}(F)$ not counted by $s_F(v_F)$. Hence, for
\[
c:=1-\frac{1}{(k-1)h},
\]
we have $0<c<1$ and
\[
\mu(F-N_F[v_F])=\mu(F)-s_F(v_F)\leq c\mu(F).
\]

We next prove by induction on $|V(F)|$ that every member of $\Upsilon(F)$ induces an $H$-free subgraph of $F$. If $F$ is $H$-free, then $\Upsilon(F)=\{V(F)\}$, and the assertion is immediate. Suppose that $F$ is not $H$-free. By induction, every member of $\Upsilon(F-v_F)$ induces an $H$-free subgraph of $F-v_F$, and hence of $F$. Now let $S\in\Upsilon(F-N_F[v_F])$. The graph $F[S]$ is $H$-free by induction, and $v_F$ has no neighbor in $S$. Thus $v_F$ is isolated in $F[S\cup\{v_F\}]$. Since $H$ has no isolated vertices, an induced copy of $H$ in $F[S\cup\{v_F\}]$ cannot contain $v_F$ and would therefore be contained in $S$, a contradiction. Hence every member of $\Upsilon(F)$ induces an $H$-free subgraph of $F$.

We now prove by induction on $|V(F)|$ that every independent set of $F$ is contained in some member of $\Upsilon(F)$. If $F$ is $H$-free, then $\Upsilon(F)=\{V(F)\}$, and the assertion is immediate. Suppose that $F$ is not $H$-free, and let $I$ be an independent set of $F$.

If $v_F\notin I$, then $I$ is an independent set of $F-v_F$. By induction, there exists $S\in\Upsilon(F-v_F)$ such that $I\subseteq S$. Since $\Upsilon(F-v_F)\subseteq\Upsilon(F)$, this gives a member of $\Upsilon(F)$ containing $I$.

Suppose that $v_F\in I$. Since $I$ is independent,
\[
I\setminus\{v_F\}\subseteq V(F)\setminus N_F[v_F]
=V(F-N_F[v_F]).
\]
Thus $I\setminus\{v_F\}$ is an independent set of $F-N_F[v_F]$. By induction, there exists $S\in\Upsilon(F-N_F[v_F])$ such that $I\setminus\{v_F\}\subseteq S$. By the definition of $\Upsilon(F)$, the set $S\cup\{v_F\}$ belongs to $\Upsilon(F)$ and contains $I$. Therefore every independent set of $F$ is contained in some member of $\Upsilon(F)$.
The preceding two properties show that $\Upsilon(F)$ is an $H$-amiable family of $F$.

It remains to bound the size of this family. Let $a:=(k-1)h$, so that $c=1-1/a$. Define $T(0):=1$ and, for every integer $m\geq1$, define
\[
T(m):=T(m-1)+T(\lfloor cm\rfloor).
\]
Since $\lfloor cm\rfloor\leq m-1$, the sequence is well-defined and strictly increasing. Let $L(F)$ be the number of leaves in the recursion tree rooted at $F$. We show by induction on $\mu(F)$ that
$
L(F)\leq T(\mu(F)).
$
If $\mu(F)=0$, then $F$ is $H$-free and $L(F)=1=T(0)$. Suppose that $m:=\mu(F)\geq1$. The two recursive subgraphs have measure at most $m-1$ and $\lfloor cm\rfloor$, respectively. Hence, by induction and the monotonicity of $T$,
\[
L(F)\leq T(m-1)+T(\lfloor cm\rfloor)=T(m).
\]
Moreover,
\[
T(m)=T(0)+\sum_{i=1}^{m}T(\lfloor ci\rfloor)
\leq1+mT(\lfloor cm\rfloor)
\leq(m+1)T(\lfloor cm\rfloor).
\]
For $m\geq1$, let $m_0:=m$ and $m_{r+1}:=\lfloor cm_r\rfloor$. Since $m_r\leq c^r m$, the first index $q$ for which $m_q=0$ satisfies $q=O_{H,k}(\log(m+1))$. Applying the preceding inequality to $m_0,\ldots,m_{q-1}$ yields
\[
T(m)\leq\prod_{r=0}^{q-1}(m_r+1)
\leq(m+1)^q
=\exp\bigl(O_{H,k}(\log^2(m+1))\bigr).
\]
Since $\mu(G)\leq\binom{n}{h}\leq n^h$, the conclusion is immediate when $n=1$, while for $n\geq2$ we have $\log(\mu(G)+1)=O_H(\log n)$. It follows that
\[
|\Upsilon(G)|\leq L(G)\leq n^{O_{H,k}(\log n)}.
\]

We finally describe an implementation attaining the same running time. Construct the binary recursion tree lazily. A node consists of a pair $(F,Z)$, where $F$ is the current induced subgraph and $Z$ is the set of vertices selected on second branches along the path from the root; the root is $(G,\emptyset)$. If $F$ is $H$-free, output $Z\cup V(F)$. Otherwise, compute $v_F$ and create the two children
\[
(F-v_F,Z)
\qquad\text{and}\qquad
(F-N_F[v_F],Z\cup\{v_F\}).
\]
An induction on the corresponding subtree shows that the set of distinct leaf outputs below a node $(F,Z)$ is
$
\{Z\cup S:S\in\Upsilon(F)\}.
$
In particular, the distinct outputs at the leaves of the whole recursion tree form precisely $\Upsilon(G)$.

At a node corresponding to $F$, the family $\mathcal{H}(F)$ and all scores can be computed in time $n^{O_H(1)}$ by enumerating the $h$-vertex subsets of $V(F)$. No internal node constructs or copies a family returned by a recursive call. Since the recursion tree is binary and has at most $T(\mu(G))$ leaves, it has at most $2T(\mu(G))-1$ nodes. Constructing each leaf output and removing repeated outputs take time polynomial in $n$ per generated set. Thus the total running time is $n^{O_{H,k}(\log n)}$. This completes the proof of~\Cref{thm:HAF}.
\end{proof}


\section{The proof of~\Cref{thm:main_1}}\label{sec:algo}

In this section, we provide a quasi-polynomial-time algorithm for \textsc{Max-Weight List $2$-Colorable Induced Subgraph} on $kP_4$-free graphs. Recall that an instance of this problem consists of a graph $G$, a weight function $w\colon V(G)\to\mathbb{Q}_{\geq 0}$, and a $2$-list assignment $\mathcal{L}\colon V(G)\to 2^{\{1,2\}}$. The task is to find a maximum-weight subset $Y\subseteq V(G)$ such that $G[Y]$ admits a coloring $\varphi\colon Y\to\{1,2\}$ with $\varphi(v)\in\mathcal{L}(v)$ for every $v\in Y$. For a $2$-list assignment $\mathcal{L}$ and $i\in\{1,2\}$, we write 
\[
V_i(\mathcal{L}):=\{v\in V(G): i\in\mathcal{L}(v)\}.
\]
For a graph $G$, a $2$-list assignment $\mathcal{L}$, and a vertex-weight function $w$ as above, we define $\lambda_2(G,\mathcal{L},w)$ to be the maximum value of $w(Y)$ over all subsets $Y\subseteq V(G)$ such that $G[Y]$ admits an $\mathcal{L}$-coloring.

We introduce the auxiliary graph used for the fixed-container problem. Let $G$ be a graph, let $\mathcal{L}\colon V(G)\to 2^{\{1,2\}}$ be a $2$-list assignment, and let $A_1,A_2\subseteq V(G)$ be vertex subsets, called \defn{containers}. We define the \defn{two-layer list auxiliary graph $\mathcal{Q}_{\mathcal{L}}(G;A_1,A_2)$} as follows. 

\begin{itemize}
\setlength\itemsep{0em}
    \item Its vertex set is $\{v^i : i\in\{1,2\},\ v\in A_i,\ \text{and } i\in\mathcal{L}(v)\}$.
    \item Two distinct vertices $u^i$ and $v^j$ are adjacent if and only if either $u=v$ and $i\neq j$, or $u\neq v$, $i=j$, and $uv\in E(G)$.
\end{itemize}

We abbreviate $\mathcal{Q}_{\mathcal{L}}(G;V(G),V(G))$ to $\mathcal{Q}_{\mathcal{L}}(G)$.
The vertices $v^1$ form the \defn{first layer} and the vertices $v^2$ form the \defn{second layer}. Edges inside one layer copy the corresponding edges of $G$, while edges of the form $v^1v^2$, called \defn{cross-edges}, forbid choosing two colors for the same original vertex. 

Given a weight function $\eta\colon V(\mathcal{Q}_{\mathcal{L}}(G;A_1,A_2))\to\mathbb{Q}$, let $\alpha_\eta(\mathcal{Q}_{\mathcal{L}}(G;A_1,A_2))$ denote the maximum $\eta$-weight of a stable set in $\mathcal{Q}_{\mathcal{L}}(G;A_1,A_2)$. For a weight function $w\colon V(G)\to\mathbb{Q}_{\geq 0}$, define the \defn{copied weight function} $\widehat{w}$ on $V(\mathcal{Q}_{\mathcal{L}}(G;A_1,A_2))$ by $\widehat{w}(v^i):=w(v)$ for every vertex $v^i$ of $\mathcal{Q}_{\mathcal{L}}(G;A_1,A_2)$.

\begin{lemma}\label{lem:list2-assignment-correspondence-vertex}
Let $G$ be a graph, let $\mathcal{L}\colon V(G)\to 2^{\{1,2\}}$ be a $2$-list assignment, let $A_1,A_2\subseteq V(G)$, and let $w\colon V(G)\to\mathbb{Q}_{\geq 0}$ be a weight function. Then stable sets of $\mathcal{Q}_{\mathcal{L}}(G;A_1,A_2)$ are in weight-preserving one-to-one correspondence with pairs $(Y,\varphi)$ such that $Y\subseteq V(G)$, $\varphi\colon Y\to\{1,2\}$ is an $\mathcal{L}$-coloring of $G[Y]$, and $\varphi^{-1}(i)\subseteq A_i$ for $i\in\{1,2\}$. In particular, $\lambda_2(G,\mathcal{L},w)=\alpha_{\widehat w}(\mathcal{Q}_{\mathcal{L}}(G))$.
\end{lemma}

\begin{proof}
Let $J$ be a stable set of $\mathcal{Q}_{\mathcal{L}}(G;A_1,A_2)$. Let
\[
Y_J:=\{v\in V(G): v^i\in J \text{ for some } i\in\{1,2\}\}.
\]
We first show that, for every $v\in Y_J$, there is a unique color $i\in\{1,2\}$ such that $v^i\in J$. Indeed, it cannot be that both $v^1$ and $v^2$ belong to $J$, since  $v^1v^2$ is an edge of $\mathcal{Q}_{\mathcal{L}}(G;A_1,A_2)$ by the definition of list auxiliary graph. We may therefore define a map $\varphi_J\colon Y_J\to\{1,2\}$ by setting $\varphi_J(v)$ to be the unique color $i$ such that $v^i\in J$.

We claim that $\varphi_J$ is an $\mathcal{L}$-coloring of $G[Y_J]$. Let $v\in Y_J$, and let $\varphi_J(v)=i$. Since $v^i$ is a vertex of $\mathcal{Q}_{\mathcal{L}}(G;A_1,A_2)$, we have $i\in\mathcal{L}(v)$ and $v\in A_i$. Thus $\varphi_J$ respects the lists and satisfies $\varphi_J^{-1}(i)\subseteq A_i$ for $i\in\{1,2\}$. Now suppose, to the contrary, that there exist distinct  vertices $u,v\in Y_J$ adjacent in $G$ and with $\varphi_J(u)=\varphi_J(v)=i$. Then $u^i,v^i\in J$. Since $uv\in E(G)$ and $u\neq v$, the definition of $\mathcal{Q}_{\mathcal{L}}(G;A_1,A_2)$ gives the edge $u^iv^i$, contradicting that $J$ is stable. Therefore $\varphi_J$ is an $\mathcal{L}$-coloring of $G[Y_J]$.

Conversely, let $(Y,\varphi)$ be a pair such that $Y\subseteq V(G)$, $\varphi\colon Y\to\{1,2\}$ is an $\mathcal{L}$-coloring of $G[Y]$, and $\varphi^{-1}(i)\subseteq A_i$ for $i\in\{1,2\}$. Define $J_{Y,\varphi}:=\{v^{\varphi(v)}:v\in Y\}$. This is a subset of $V(\mathcal{Q}_{\mathcal{L}}(G;A_1,A_2))$, because for every $v\in Y$, the color $\varphi(v)$ belongs to $\mathcal{L}(v)$ and $v\in A_{\varphi(v)}$.

We now claim that $J_{Y,\varphi}$ is a stable set of $\mathcal{Q}_{\mathcal{L}}(G;A_1,A_2)$. First, there is no edge of the form $v^1v^2$ with both endpoints in $J_{Y,\varphi}$, because the construction of $J_{Y,\varphi}$ selects exactly one copy of each vertex $v\in Y$. Second, suppose that $u^i$ and $v^i$ are two distinct vertices of $J_{Y,\varphi}$. Then $\varphi(u)=\varphi(v)=i$. Since $\varphi$ is a coloring of $G[Y]$, we have $uv\notin E(G)$. Hence, by definition of $\mathcal{Q}_{\mathcal{L}}(G;A_1,A_2)$, $u^i$ and $v^i$ are not adjacent in $\mathcal{Q}_{\mathcal{L}}(G;A_1,A_2)$. We conclude that $J_{Y,\varphi}$ is indeed stable.

The two constructions are inverse to each other. Starting with a stable set $J$, the set obtained from $(Y_J,\varphi_J)$ is precisely $\{v^{\varphi_J(v)}:v\in Y_J\}=J$. Conversely, starting with a feasible pair $(Y,\varphi)$, the coloring recovered from $J_{Y,\varphi}$ has domain $Y$ and assigns each vertex $v\in Y$ the color $\varphi(v)$.

Finally, the correspondence is weight-preserving for the copied weight function. For every stable set $J$, we have $\widehat{w}(J)=\sum_{v^i\in J}w(v)=w(Y_J)$. Therefore, in the special case $A_1=A_2=V(G)$, maximizing $\widehat{w}(J)$ over stable sets of $\mathcal{Q}_{\mathcal{L}}(G)$ is exactly the same as maximizing $w(Y)$ over all induced subgraphs $G[Y]$ admitting an $\mathcal{L}$-coloring. Hence $\lambda_2(G,\mathcal{L},w)=\alpha_{\widehat w}(\mathcal{Q}_{\mathcal{L}}(G))$.
\end{proof}

Recall that a graph $G$ is \defn{perfect} if every induced subgraph $F$ of $G$ satisfies $\chi(F)=\omega(F)$. A \defn{hole} of $G$ is an induced subgraph of $G$ which is a cycle of length at least four, an \defn{antihole} of $G$ is an induced subgraph of $G$ whose complement is a hole of $\overline{G}$, and a graph is \defn{Berge} if it contains no odd hole and no odd antihole. 

Our proof of \Cref{thm:main_1} will use two deep results in structural and algorithmic graph theory. We will use the Strong Perfect Graph Theorem of \citet*{chudnovsky2006strong}, stating that a graph is Berge if and only if it is perfect, and the following result of \citet*{grotschel1981ellipsoid} (see also~\citep{Sch}).

\begin{theorem}[\citet*{grotschel1981ellipsoid}]\label{thm:IS-in-perfectgraphs}
\textsc{Max-Weight Independent Set} can be solved in polynomial time on perfect graphs.
\end{theorem}

The crucial observation is that, when both containers induce $P_4$-free graphs, the list auxiliary graph is perfect.

\begin{lemma}\label{lem:list2-aux-perfect}
Let $G$ be a graph, let $\mathcal{L}\colon V(G)\to2^{\{1,2\}}$ be a $2$-list assignment, and let $A_1,A_2\subseteq V(G)$. If $G[A_1]$ and $G[A_2]$ are $P_4$-free, then $\mathcal{Q}_{\mathcal{L}}(G;A_1,A_2)$ is perfect.
\end{lemma}

\begin{proof}
Let
\[
\mathcal{Q}:=\mathcal{Q}_{\mathcal{L}}(G;A_1,A_2).
\]
For $i\in\{1,2\}$, denote by $W_i$ the $i$-th layer, that is,
\[
W_i:= \{v^i : v\in A_i \text{ and } i\in\mathcal{L}(v)\}.
\]
Thus $V(\mathcal{Q})=W_1\cup W_2$, and $W_1,W_2$ are disjoint. The graph $\mathcal{Q}[W_i]$ is isomorphic to the induced subgraph $G[A_i\cap V_i(\mathcal{L})]$. Since $G[A_i]$ is $P_4$-free and $A_i\cap V_i(\mathcal{L})\subseteq A_i$, the graph $\mathcal{Q}[W_i]$ is $P_4$-free for each $i\in\{1,2\}$.
By the Strong Perfect Graph Theorem, it is enough to prove that
 $\mathcal{Q}$ is Berge. 

Suppose first that $\mathcal{Q}$ contains an odd hole $C$. If $C$ is contained entirely in $W_i$ for some $i\in\{1,2\}$, then $\mathcal{Q}[W_i]$ contains an induced cycle of odd length at least five. Taking four consecutive vertices of this induced cycle gives an induced $P_4$ in $\mathcal{Q}[W_i]$, contradicting that $\mathcal{Q}[W_i]$ is $P_4$-free. Therefore $C$ contains vertices from both layers.

The only edges between the two layers are edges of the form $v^1v^2$, i.e., cross-edges. Consequently, every passage in $C$ from one layer to the other uses a cross-edge. Let $r$ be the number of cross-edges of $C$. Since a cycle must return to the layer from which it starts, $r$ is even. Also $r\geq 2$, because $C$ uses both layers.

Delete now from $C$ all cross-edges. The remaining components are paths, each contained entirely in one layer. We call these paths the \defn{layer-segments} of $C$. Since every vertex of $\mathcal{Q}$ is incident with at most one cross-edge, two distinct cross-edges of $C$ cannot be consecutive on $C$. Hence every layer-segment has length at least one. Moreover, every layer-segment is an induced path inside its layer, because any chord of such a path inside the layer would be a chord of the hole $C$. Since both layers are $P_4$-free, no layer-segment has length at least three. Therefore, every layer-segment has length one or two.

Suppose first that $r=2$. Then $C$ consists of two cross-edges, say $x^1x^2$ and $y^1y^2$, together with one layer-segment joining $x^1$ to $y^1$ in $W_1$ and one layer-segment joining $x^2$ to $y^2$ in $W_2$. Here $x\neq y$, because the two cross-edges are distinct. The segment between $x^1$ and $y^1$ has length one if and only if $xy\in E(G)$ and, by the previous paragraph, it has length two if and only if $xy \notin E(G)$. The same statement holds for the segment between $x^2$ and $y^2$. Since the adjacency of $x$ and $y$ in $G$ is independent of the layer, the two layer-segments have the same length. Therefore the length of $C$ is either $2+1+1=4$ or $2+2+2=6$. In either case, $C$ is not an odd hole, a contradiction.

Suppose finally that $r\geq 4$. We claim that no layer-segment has length one. Suppose, to the contrary, that some layer-segment has length one. By symmetry, suppose that this segment is the edge $x^1y^1$ in $W_1$, where $x^1x^2$ and $y^1y^2$ are cross-edges of $C$. Then $xy\in E(G)$, and hence $x^2y^2$ is an edge of $\mathcal{Q}[W_2]$. Since $r\geq 4$, the $x^2$--$y^2$ arc of $C$ different from the path $x^2x^1y^1y^2$ contains at least one further cross-edge, and hence contains an internal vertex. Therefore this arc has length at least two, so $x^2$ and $y^2$ are not consecutive on $C$. Thus the edge $x^2y^2$ is a chord of $C$, contradicting that $C$ is induced. This proves that every layer-segment has length two. But since there are exactly $r$ layer-segments, the length of $C$ is $r+2r=3r$. Since $r$ is even, this length is even, contradicting that $C$ is an odd hole. Thus $\mathcal{Q}$ contains no odd hole.

It remains to argue that $\mathcal{Q}$ contains no odd antihole. Since a $5$-antihole is also a $5$-hole, and we have already shown that $\mathcal{Q}$ contains no odd hole, it is enough to exclude odd antiholes of length at least seven. Suppose, to the contrary, that $\mathcal{Q}$ contains an odd antihole $C$ of length at least seven. If $C$ is contained entirely in one layer, say $W_i$, then $\mathcal{Q}[W_i]$ contains an induced odd antihole of length at least seven. The complement of $\mathcal{Q}[V(C)]$ is then an induced odd cycle of length at least seven. Taking four consecutive vertices of this induced cycle gives an induced $P_4$ in the complement of $\mathcal{Q}[V(C)]$. Since $P_4$ is self-complementary, the same four vertices induce a $P_4$ in $\mathcal{Q}[V(C)]$, and hence in $\mathcal{Q}[W_i]$. This contradicts that $\mathcal{Q}[W_i]$ is $P_4$-free. Thus $C$ meets both layers.

Let $a:=|V(C)\cap W_1|$ and $b:=|V(C)\cap W_2|$. Since $C$ is an antihole of length at least seven, every vertex of $C$ has exactly two non-neighbors inside $V(C)$. Take a vertex $x^1\in V(C)\cap W_1$. The only possible neighbor of $x^1$ in the other layer $W_2$ is $x^2$. Hence $x^1$ is adjacent to at most one vertex of $V(C)\cap W_2$, and so it is non-adjacent to at least $b-1$ vertices of $V(C)\cap W_2$. Since $x^1$ has exactly two non-neighbors in $V(C)$, we get $b-1\leq 2$, and hence $b\leq 3$. By the symmetric argument applied to a vertex of $V(C)\cap W_2$, we get $a\leq 3$. Therefore $|V(C)|=a+b\leq 6$, contradicting that $C$ has length at least seven. Thus $\mathcal{Q}$ contains no odd antihole.
\end{proof}

We now give the algorithm for \textsc{Max-Weight List $2$-Colorable Induced Subgraph} on $kP_4$-free graphs.

\begin{algorithm}[]
\caption{$\Lambda_2(G,k,\mathcal{L},w)$}
\label{alg:list2}
\begin{algorithmic}[1]
\Require A $kP_4$-free graph $G$, a $2$-list assignment $\mathcal{L}\colon V(G)\to2^{\{1,2\}}$, and a weight function $w\colon V(G)\to\mathbb{Q}_{\geq0}$
\Ensure A maximum-weight set $Y\subseteq V(G)$ such that $G[Y]$ admits an $\mathcal{L}$-coloring

\If{$k=1$} \label{alg:list2:k1-test}
    \State $\mathcal{F}\gets \{V(G)\}$ \label{alg:list2:k1-family}
\Else
    \State Compute a $P_4$-amiable family $\mathcal{F}$ of $G$ using \Cref{thm:HAF} \label{alg:list2:haf-family}
\EndIf

\State $Y_{\mathrm{best}}\gets \emptyset$ \label{alg:list2:init}

\ForAll{$(A_1,A_2)\in \mathcal{F}\times \mathcal{F}$} \label{alg:list2:loop}
    \State Construct $\mathcal{Q}_{A_1,A_2}\gets \mathcal{Q}_{\mathcal{L}}(G;A_1,A_2)$ \label{alg:list2:construct-q}
    \State Define $\widehat{w}(v^i)\gets w(v)$ for every $v^i\in V(\mathcal{Q}_{A_1,A_2})$ \label{alg:list2:weights}
    \State Compute a maximum-weight stable set $J$ of $\mathcal{Q}_{A_1,A_2}$ with respect to $\widehat{w}$ using \Cref{thm:IS-in-perfectgraphs} \label{alg:list2:max-stable}
    \State $Y\gets \{v\in V(G): v^i\in J \text{ for some } i\in\{1,2\}\}$ \label{alg:list2:define-Y}
    \If{$w(Y)>w(Y_{\mathrm{best}})$} \label{alg:list2:compare}
        \State $Y_{\mathrm{best}}\gets Y$ \label{alg:list2:update}
    \EndIf
\EndFor

\State \Return $Y_{\mathrm{best}}$ \label{alg:list2:return}
\end{algorithmic}
\end{algorithm}

\begin{theorem}\label{thm:list2-kp4}
For every fixed $k\in\mathbb{N}$, \textsc{Max-Weight List $2$-Colorable Induced Subgraph} can be solved in quasi-polynomial time on $kP_4$-free graphs. More precisely, given an $n$-vertex $kP_4$-free graph $G$, a $2$-list assignment $\mathcal{L}\colon V(G)\to2^{\{1,2\}}$, and a weight function $w\colon V(G)\to\mathbb{Q}_{\geq0}$, Algorithm~\ref{alg:list2} returns an optimum solution in time $n^{O_k(\log n)}$.
\end{theorem}

\begin{proof}
Fix $k\in\mathbb{N}$, and let $G$ be an $n$-vertex $kP_4$-free graph with $2$-list assignment $\mathcal{L}$ and weight function $w$.
Let $\mathcal{F}$ be the family computed by Algorithm~\ref{alg:list2} in lines~\ref{alg:list2:k1-test}--\ref{alg:list2:haf-family}. We first establish two properties of $\mathcal{F}$ used throughout the proof. We claim that every member of $\mathcal{F}$ induces a $P_4$-free graph in $G$, and every stable set of $G$ is contained in some member of $\mathcal{F}$. If $k=1$, then $G$ is $P_4$-free, and line~\ref{alg:list2:k1-family} sets $\mathcal{F}=\{V(G)\}$; both properties are immediate. If $k\geq2$, since $P_4$ has no isolated vertices, \Cref{thm:HAF} applies with $H=P_4$. Thus line~\ref{alg:list2:haf-family} computes a $P_4$-amiable family, and the two properties follow from the definition of $P_4$-amiability. Note that $|\mathcal{F}|=n^{O_k(\log n)}$, and the family is computed in time $n^{O_k(\log n)}$ by \Cref{thm:HAF}.

We next verify that every iteration of the loop in line~\ref{alg:list2:loop} is well-defined. Let $(A_1,A_2)$ be the ordered pair considered in an arbitrary iteration. By the previous paragraph, $G[A_1]$ and $G[A_2]$ are $P_4$-free. Hence \Cref{lem:list2-aux-perfect} implies that the graph $\mathcal{Q}_{A_1,A_2}$ constructed in line~\ref{alg:list2:construct-q} is perfect, and so the maximum-weight stable set required in line~\ref{alg:list2:max-stable} can indeed be computed in polynomial time.

We now prove that Algorithm~\ref{alg:list2} returns a feasible set. Let $N:=|\mathcal{F}\times\mathcal{F}|$, and fix the order in which the loop in line~\ref{alg:list2:loop} processes the ordered pairs. For $t\in\{1,\dots,N\}$, let $(A_1^t,A_2^t)$ be the ordered pair processed in the $t$-th iteration. Let $J_t$ be the stable set computed in line~\ref{alg:list2:max-stable} during this iteration, and let $Y_t:=\{v\in V(G):v^i\in J_t \text{ for some } i\in\{1,2\}\}$ be the set constructed in line~\ref{alg:list2:define-Y}. By the proof of \Cref{lem:list2-assignment-correspondence-vertex}, applied with $A_1=A_1^t$ and $A_2=A_2^t$, the stable set $J_t$ corresponds to the pair $(Y_t,\varphi_t)$, where $\varphi_t$ is an $\mathcal{L}$-coloring of $G[Y_t]$. Therefore each $Y_t$ is feasible for the input instance. The initial set $\emptyset$ in line~\ref{alg:list2:init} is also feasible.

Note that, after the loop terminates, the set returned in line~\ref{alg:list2:return} is feasible and has maximum weight among all candidates produced by the loop.

It remains to prove that one of the candidates has optimum weight. Let $(Y^\star,\varphi^\star)$ be an optimum solution of the input instance. For $i\in\{1,2\}$, set $I_i^\star:=(\varphi^\star)^{-1}(i)$. Since $\varphi^\star$ is a coloring of $G[Y^\star]$, the sets $I_1^\star$ and $I_2^\star$ are stable sets of $G$. By the covering property of $\mathcal{F}$, choose $A_1^\star,A_2^\star\in\mathcal{F}$ such that $I_i^\star\subseteq A_i^\star$ for $i\in\{1,2\}$. The loop in line~\ref{alg:list2:loop} enumerates every ordered pair of members of $\mathcal{F}$, so there is an iteration, say the $t^\star$-th iteration, in which $(A_1^{t^\star},A_2^{t^\star})=(A_1^\star,A_2^\star)$.

Consider the graph $\mathcal{Q}_{A_1^\star,A_2^\star}$ constructed in line~\ref{alg:list2:construct-q} during this iteration, and define $J^\star:=\{v^{\varphi^\star(v)}:v\in Y^\star\}$. For every $v\in Y^\star$, the vertex $v^{\varphi^\star(v)}$ belongs to $\mathcal{Q}_{A_1^\star,A_2^\star}$: indeed, $\varphi^\star(v)\in\mathcal{L}(v)$ because $\varphi^\star$ respects the lists, and $v\in I_{\varphi^\star(v)}^\star\subseteq A_{\varphi^\star(v)}^\star$ by the choice of the containers. By the proof of \Cref{lem:list2-assignment-correspondence-vertex}, $J^\star$ is a stable set of $\mathcal{Q}_{A_1^\star,A_2^\star}$ and its copied weight is $\widehat w(J^\star)=w(Y^\star)$.

Let $J_{t^\star}$ be the maximum-weight stable set computed in line~\ref{alg:list2:max-stable} during the $t^\star$-th iteration. Since $J^\star$ is feasible for this maximum-weight stable set computation, we have $\widehat w(J_{t^\star})\geq \widehat w(J^\star)=w(Y^\star)$. Let $Y_{t^\star}$ be the set constructed from $J_{t^\star}$ in line~\ref{alg:list2:define-Y}. By the proof of \Cref{lem:list2-assignment-correspondence-vertex}, $Y_{t^\star}$ is feasible and satisfies $w(Y_{t^\star})=\widehat w(J_{t^\star})$. Therefore $w(Y_{t^\star})\geq w(Y^\star)$. Then the set returned in line~\ref{alg:list2:return} has weight at least $w(Y_{t^\star}) \geq w(Y^\star)$ and hence exactly $w(Y^\star)$ by the optimality of $Y^\star$. Therefore Algorithm~\ref{alg:list2} returns an optimum solution.

It remains to bound the running time. The family $\mathcal{F}$ is computed in time $n^{O_k(\log n)}$ and has size $n^{O_k(\log n)}$. Hence the loop in line~\ref{alg:list2:loop} has $|\mathcal{F}|^2=n^{O_k(\log n)}$ iterations. In each iteration, the graph $\mathcal{Q}_{\mathcal{L}}(G;A_1,A_2)$ constructed in line~\ref{alg:list2:construct-q} has at most $2n$ vertices and can be constructed in polynomial time. The copied weight function in line~\ref{alg:list2:weights} is constructed in polynomial time. By \Cref{lem:list2-aux-perfect}, the list auxiliary graph is perfect, and hence line~\ref{alg:list2:max-stable} takes time polynomial in the size of the list auxiliary graph. The construction of $Y$ in line~\ref{alg:list2:define-Y}, the comparison in line~\ref{alg:list2:compare}, and the possible update in line~\ref{alg:list2:update} are polynomial-time operations. Thus the total running time is $n^{O_k(\log n)}$. This concludes the proof of~\Cref{thm:list2-kp4}.
\end{proof}

\section{Concluding remarks}\label{sec:conclusion}

We showed that, for every $k \in \mathbb{N}$, \textsc{Max-Weight List $2$-Colorable Induced Subgraph} is solvable in quasi-polynomial time on $kP_4$-free graphs. As \textsc{Odd Cycle Transversal}
is a special case, this yields the same conclusion for \textsc{Odd Cycle Transversal}. Paired with known results from the literature, this resolves in particular the open problem posed by \citet*{ALLSS24} of obtaining a complete (\textsf{QP} vs \textsf{NP}-hard)-dichotomy for \textsc{Odd Cycle Transversal} on $H$-free graphs. 

A first obvious open problem is to determine whether our quasi-polynomial-time algorithms can be improved to polynomial time: Do \textsc{Odd Cycle Transversal} and \textsc{Max-Weight List $2$-Colorable Induced Subgraph} admit a polynomial-time algorithm on $kP_4$-free graphs? A positive answer would provide a complete (\textsf{P} vs \textsf{NP}-hard)-dichotomy for both problems on $H$-free graphs (see \Cref{thm:literature-oct}). This also raises the natural question of whether every $kP_4$-free graph admits a $P_4$-amiable family of polynomial size and computable in polynomial
time. 

Another natural direction is to obtain similar dichotomies for \textsc{Max-Weight List $r$-Colorable Induced Subgraph} on $H$-free graphs in the remaining cases $r \in  \{3, 4\}$. Similarly to the case $r = 2$, it remains to consider linear forests $H$ of the form $H = k_4P_4 + k_3P_3 + k_2P_2 + k_1P_1$, with $k_4 \geq 1$ and either $k_3 + k_2 \geq 1$ or $k_4 \geq 2$. Note that, when $r \in \{3,4\}$, even the behavior of the decision version \textsc{List $r$-Coloring} remains rather obscure. Indeed, even though a quasi-polynomial-time algorithm for \textsc{List $3$-Coloring} on $P_t$-free graphs for every $t\in \mathbb{N}$ has been recently obtained (see \citep[Theorem~1(a)]{OR21}), thus implying an analogous result on $kP_4$-free graphs, to the best of our knowledge, no quasi-polynomial-time algorithm for \textsc{List $4$-Coloring} on $kP_4$-free graphs is known. Moreover, it seems challenging even to provide a polynomial-time algorithm for \textsc{$3$-Coloring} on $2P_4$-free graphs \citep{JKMNP22}.

Another related direction is to extend \Cref{thm:main_1} from $kP_4$-free
graphs to graphs of bounded induced-$P_4$-packing treewidth~\cite{nikabadi2026induced}. The local
container lemma of~\cite{nikabadi2026induced} produces residual instances containing no induced $P_4$ that intersects a distinguished bag. However,
our~\Cref{lem:list2-aux-perfect} uses the stronger property that both layers
of the auxiliary graph are globally $P_4$-free. Thus such an extension
would require a local counterpart of the perfectness argument, or a
different method for solving the resulting residual instances.

\bibliographystyle{abbrvnat}
\bibliography{ref}

\end{document}